\documentclass[11pt]{article}

\usepackage{amsfonts,amssymb,amsmath,latexsym,xfrac,ae,aecompl}

\usepackage{graphicx}

\newcommand{\F}{\vspace*{\smallskipamount}}

\newcommand{\FFF}{\vspace*{\bigskipamount}}

\newcommand{\BBB}{\vspace*{-\bigskipamount}}

\newcommand{\cD}{\mathcal{D}}

\newcommand{\cO}{\mathcal{O}}

\newcommand{\mE}{\mathbb{E}}

\newcommand{\Paragraph}[1]{\BBB\paragraph{#1}}
\newcommand{\remove}[1]{}

\newlength{\pagewidth}
\newlength{\captionwidth}
\newcommand{\qed}{\hfill $\square$ \smallbreak}
\newenvironment{proof}{\noindent{\bf Proof:}}{\qed}

\newtheorem{theorem}{Theorem}
\newtheorem{lemma}{Lemma}

\newtheorem{proposition}{Proposition}

\begin{document}

\baselineskip           	3ex
\parskip                	1ex

\title{		Naming a Channel with Beeps \footnotemark[1] \FFF\FFF\FFF}

\author{	Bogdan S. Chlebus\,\footnotemark[2]   	
		\and
		Gianluca De Marco\,\footnotemark[3]
		\and
		Muhammed Talo\,\footnotemark[4]}

\footnotetext[2]{This paper was published as~\cite{ChlebusDT-FI17}.}

\footnotetext[2]{Department of Computer Science and Engineering, 
			University of Colorado Denver, 
			Denver, Colorado 80217, USA.
			The work of this author was supported by the National Science Foundation under Grant 1016847.}

\footnotetext[3]{Dipartimento di Informatica,
                  Universit\`a degli Studi di Salerno,
                  Fisciano, 84084 Salerno, Italy.}
                  
\footnotetext[4]{Bilgisayar M\"{u}hendisli\u{g}i,
		Munzur \"{U}niversitesi,
		62000 Tunceli, Turkey.}

\date{}

\maketitle

\vfill


\begin{abstract}
We consider a communication channel in which the only available mode of communication  is transmitting beeps.
A beep transmitted by a station attached to the channel reaches all the other stations instantaneously.
Stations are anonymous, in that they do not have any individual identifiers.
The algorithmic goal is to assign names to the stations in such a manner that the names make a contiguous segment of positive integers starting from~$1$.
We develop a Las Vegas naming algorithm, for the case when the number of stations~$n$ is known, and a Monte Carlo algorithm, for the case when the number of stations~$n$ is not known.
The given randomized algorithms are provably optimal with respect to the expected time $\cO(n\log n)$, the expected number  of used random bits $\cO(n\log n)$, and the probability of error.

\vfill

~

\end{abstract}

\vfill

\thispagestyle{empty}

\setcounter{page}{0}

\newpage


\section{Introduction}

\label{sec:introduction}

The considered model of communication is anonymous channel with beeping.
There are a number of stations attached to the channel that are devoid of any identifying features they can refer to during communication.
The goal is to have the stations assign names to themselves by executing a distributed communication algorithm.

The stations communicate among themselves in synchronous rounds.
All of them start together in the same round.
The channel provides a binary feedback to all the attached stations: when no stations transmit then nothing is sensed on the communication medium, and when some station does transmit then every station detects a beep.

A beeping channel resembles multiple-access channels, in that it can be interpreted as a single-hop radio network.
The difference between the two models is in the feedback they provide.
Multiple access channels come in two variants: with and without collision detection.
A multiple access channel without collision detection provides the following binary feedback: when exactly one station transmits then the transmitted message is heard by every station, and otherwise, when either no station transmits or multiple stations do, then this results in each station sensing silence.
A multiple access channel with collision detection provides the following ternary feedback: silence occurs when no station transmits, a message is heard when exactly one station transmits, and collision occurs when multiple stations transmit simultaneously, which prevents any of the transmitted messages to be heard and is sensed as distinct from silence.

A channel with beeping has its communication capabilities restricted only to carrier sensing, without even the functionality of transmitting specific bits as messages.
The only apparent mode of exchanging information on such a synchronous channel with beeping is to suitably encode it by sequences of beeps and silences.

Modeling communication by a mechanism as limited as beeping has been motivated by diverse aspects of communication and distributed computing.
Beeping provides a detection of collision on a transmitting medium by sensing it.
Communication only by carrier sensing can be placed in a general context of investigating wireless communication on the physical level and modeling interference of concurrent transmissions, of which the signal-to-interference-plus-noise ratio (SINR) model is among the most popular and well studied; see~\cite{GoussevskaiaPW10-survey,JurdzinskiK-encyclopedia,SchmidW06}.
Beeping is then a very limited mode of wireless communication, with feedback in the form of either a simplest signal, which is identical to interference, or lack thereof.

Another motivation comes from biological systems, in which agents exchange information in a distributed manner, while the environment severely restricts how such agents communicate; see~\cite{AfekABHBB11,JeavonsSX2016,NavlakhaB2014}.

Finally, communication with beeps belongs to the area of distributed computing by weak devices, where the involved agents have restricted computational and communication capabilities.
In this context, the devices are modeled as finite-state machines that communicate asynchronously by exchanging states or messages from a finite alphabet.
Examples of this approach include the ``population-protocols'' model introduced by Angluin et at.~\cite{AngluinADFP2006}, (see also~\cite{AngluinAFJ08,AspnesR07,MichailCS11}), and the ``stone-age'' distributed computing model proposed by Emek and Wattenhofer~\cite{EmekW13}.

\Paragraph{Previous work.}

The model of  communication by  discrete beeping was introduced by Cornejo and Kuhn~\cite{CornejoK10}, who considered a general-topology wireless network in which nodes use only carrier sensing to communicate, and developed algorithms for node coloring.
They were inspired by ``continuous'' beeping studied by Degesys et al.~\cite{DegesysRPN07} and Motskin et al.~\cite{MotskinRSG09}, and by the implementation of coordination by carrier sensing given by Flury and Wattenhofer~\cite{FluryW10}.

Afek et al.~\cite{AfekABCHK13} considered the problem to find a maximal independent set of nodes in a distributed manner when the nodes can only beep, under additional assumptions regarding the knowledge of the size of the network, waking up the network by beeps, collision detection among concurrent beeps, and synchrony.
Brandes et al.~\cite{BrandesKKPW16} studied the problem of randomly estimating the number of nodes attached to a single-hop beeping network.  
Czumaj and Davies~\cite{CzumajD15} approached systematically the tasks of deterministic  broadcasting, gossiping, and multi-broadcasting on the bit level in general-topology symmetric beeping networks.
In a related work, Hounkanli and Pelc~\cite{HounkanliP16} studied deterministic broadcasting in asynchronous beeping networks of general topology with various levels of knowledge about the network.
F{\"{o}}rster et al.~\cite{ForsterSW14} considered leader election by deterministic algorithms in general multi-hop networks with beeping.
Gilbert and Newport~\cite{GilbertN15} studied the efficiency of leader election in a beeping single-hop channel when nodes are state machines of constant size with a specific precision of randomized state transitions.
Huang and Moscibroda~\cite{HuangM13} considered the problems of identifying subsets of stations connected to a beeping channel and compared their complexity to those on multiple-access channels.
Yu et al.~\cite{YuJYLC15} considered the problem of constructing a minimum dominating set in networks with beeping.

\Paragraph{Related work.}

A beeping channel is related to multiple access channels and single-hop radio networks~\cite{Chlebus-randomized-radio-chapter-2001}.
The similar task to the one considered in this paper, to have all stations transmit successfully, once by each station among the $n$ stations, can be accomplished in the expected time~$\cO(n)$ on a multiple access channel.
This variant of the ``contention resolution problem'' for such channels has been the topic of continuous interest, since multiple access channels started to be investigated
\cite{AntaMM13,Capetanakis79,CidonS88,DeMarcoK10,DeMarcoK17,DeMarcoK15,GreenbergFL87,Hayes78,TsybakovM1978}.
In view of the lower bound $\Omega(n\log n)$ on time of naming for beeping channels, given in this paper, resolving contention in the model of a beeping channel is demonstrably more costly, in that it involves the additional logarithmic multiplicative overhead.
This can be attributed to the key difference between the two models, in that a multiple access channel differentiates between single and multiple transmissions by its feedback, while a beeping channel returns a beep in both cases. 

The requirements of the contention resolution problem can be weakened  such that we want to hear just some node as soon as possible; this problem is called \emph{waking up} the network when nodes start spontaneously rather than coordinated from the start.
This problem was studied initially for multiple access channels \cite{DeMarcoPS07,GasieniecPP-JDM01}, and further for radio networks~\cite{ChlebusGKR-ICALP05} and multi-channel networks~\cite{ChlebusDK16}.

Networks of nodes communicating by beeps share common features with radio networks with collision detection.
Ghaffari and Haeupler~\cite{GhaffariH13} gave efficient leader election algorithm by treating collision detection as ``beeping'' and transmitting messages as bit strings.
Their approach by way of ``beep waves'' was adopted to broadcasting in networks with beeping by Czumaj and Davies~\cite{CzumajD15}.
In a related work, Ghaffari et al.~\cite{GhaffariHK15} developed randomized broadcasting and multi-broadcasting in radio networks with collision detection.

Communication in anonymous networks was first considered by Angluin~\cite{Angluin80}, who showed that randomization was necessary to name nodes in suitably symmetric communication settings.
The work that followed was either on specific network topologies, like rings or hypercubes, or on problems in general message-passing systems.
The early work on anonymous rings was done by Attiya et al.~\cite{AttiyaSW88}, while anonymous hypercubes were studied first by Kranakis and Krizanc~\cite{KranakisK97}.
Algorithmic problems in anonymous networks of general topologies were considered by 
Afek and Matias~\cite{AfekM94} and Schieber and Snir~\cite{SchieberS94}.
Angluin et al.~\cite{AngluinAFJ08} considered self-stabilizing protocols for anonymous asynchronous agents.

Anonymous systems of asynchronous processes communicating by shared memory were studied by Lipton and Park~\cite{LiptonP-1990}, E{\u{g}}ecio{\u{g}}lu and Singh~\cite{EgeciogluS94},  Panconesi et al.~\cite{PanconesiPTV97}, Kutten et al.~\cite{KuttenOP00}, and Buhrman et al.~\cite{BuhrmanPSV-2006}, among others.
Chlebus et al.~\cite{ChlebusMT15,ChlebusDT-OPODIS17} considered synchronous anonymous distributed computing with shared read-write memory.

Impossibility results and lower bounds are surveyed by Attiya and Ellen~\cite{Attiya-Ellen-book-2014} and Fich and Ruppert~\cite{FichR03}.
Yao's minimax principle was formulated by Yao~\cite{Yao77}, with an alternative direct proof given by Fich et al.~\cite{FichHRW85}; see also the book by Motwani and Raghavan~\cite{MotwaniR95} for a presentation of applications.

\section{Technical preliminaries}

\label{sec:preliminaries}

A beeping channel is a network consisting of some $n$ stations connected to a communication medium.
The channel provides a very limited feedback: when no stations transmit at a point in time then nothing is sensed on the communication medium, and when some station does transmit then every station detects a beep.

We consider synchronous  beeping channels, in the sense that an execution of a communication algorithm is partitioned into consecutive rounds.
All the stations start execution together.
In each round, a station may either beep or pause.
When some station beeps in a round, then each station hears the beep in this round, otherwise all the stations receive silence as feedback in this round.
When multiple stations beep together in a round then this creates a \emph{collision}.

We say that a parameter of a communication network is \emph{known} when it can be used in codes of algorithms.
The relevant parameter used in this paper is the number of stations~$n$.
We consider two cases, in which either $n$ is known or it is not, respectively.

The stations attached to the channel are anonymous in that they do not have any predetermined identifiers that can be referred to when executing a communication algorithm.
The stations are distinguishable by an external observer, and we may refer to such an observer's ``names'' when referring to stations when we want to distinguish the stations among themselves. 

Randomized algorithms use random bits, understood as outcomes of tosses of a fair coin.
All different random bits used by our algorithms are considered stochastically independent from one another.

Randomized algorithms are often categorized as Las Vegas or Monte Carlo, see~\cite{MotwaniR95}.
We use these terms in the following meaning.
An algorithm is \emph{Las Vegas} when it terminates almost surely and it returns a correct output upon termination.
An algorithm is \emph{Monte Carlo} when it terminates almost surely but an error may occur when the algorithm terminates; the probability of error converges to zero with the number of communicating agents converging to infinity.
An ``error'' occurs when the output generated by the algorithm does not satisfy all the desired properties, as specified in the correctness requirements.
Our naming algorithms have as  their goal to assign unique identifiers to the stations, moreover we want names to be integers in the contiguous range $\{1,2,\ldots,n\}$; we denote this rage as~$[n]$.
The Monte Carlo naming algorithm that we develop  assigns names that completely fill in an interval of integers of the form $[k]$, for $k\le n$.
If it happens that $k<n$ then this means that there are duplicate integers assigned as names, which is the only form of error that can possibly occur.

We will use a procedure to detect collisions, called \textsc{Detect-Collision}, whose pseudocode is in Figure~\ref{fig:procedure-collision}.
The procedure is executed by a group of stations, and they all start their executions simultaneously.
The procedure takes two rounds.
Each of the participating stations simulates the toss of a fair coin, with the outcomes mutually independent among the participating stations.
Determined by the outcome of a toss, a station beeps either in the first round or the second one of the allocated pair of rounds.
A collision is detected only when two consecutive beeps are heard.


\begin{figure}[t]
\rule{\textwidth}{0.75pt}

\F 
\textbf{Procedure} \textsc{Detect-Collision}

\rule{\textwidth}{0.75pt}
\begin{center}
\begin{minipage}{\pagewidth}
\begin{description}
\item[\tt toss$_v$] $\gets$  outcome of a random coin toss\ 
\item[\tt if] $\texttt{toss}_v =  \texttt{heads}$ \hfill /$\ast$ first round $\ast$/
\begin{description}
\item[\texttt{then}] \texttt{beep else pause}   
\end{description}
\item[\tt if] $\texttt{toss}_v =  \texttt{tails}$ \hfill /$\ast$ second round $\ast$/
\begin{description}
\item[\texttt{then}] \texttt{beep else pause}   
\end{description}

\item[\tt return] (a beep was heard in each of the two rounds)
\end{description}
\end{minipage}
\FFF

\rule{\textwidth}{0.75pt}

\parbox{\captionwidth}{\caption{\label{fig:procedure-collision}
A pseudocode for a station $v$.
The procedure takes two rounds to be executed.
It detects a collision and returns ``true'' when a beep is heard in each of the two rounds, otherwise it does not detect a collision and returns ``false.''
}}
\end{center}
\end{figure}

\begin{lemma}
\label{lem:detect-collision}

If $k$ stations perform $m$ time-disjoint calls of procedure \textsc{Detect-Collision}, with each station participating in exactly one call, then collision is not detected in any of these calls with probability $2^{-k+m}$.
\end{lemma}

\begin{proof}
Consider a call of \textsc{Detect-Collision} performed concurrently by $i$ stations, for $i\ge 1$.
We argue by ``deferred decisions.'' 
One of these stations tosses a coin and determines its outcome~$X$.
The other $i-1$ stations participating concurrently in this call also toss their coins.
Here we have $i-1\ge 0$, so it could be the case that there is no such a station; if $i=1$ then collision cannot be detected.
The only possibility not to detect a collision in the case of $i-1>0$ is for each among  these   $i-1$ stations to also produce~$X$.
This happens with probability $2^{-i+1}$ in one call.
The probability of producing only outcome \texttt{false} during the $m$ calls is the product of these probabilities.
When we multiply them out over $m$ instances of the procedure being performed, then the outcome is $2^{-k+m}$, because the numbers~$i$ sum up to~$k$ and the number of factors is~$m$.
\end{proof}

\Paragraph{Pseudocode conventions and notations.}

We present algorithms formally  using pseudocodes, like the one in Figure~\ref{fig:procedure-collision}.
The convention we apply is that all the stations stay synchronized to execute the same line of code in each  round.
This means that if a conditional statement is executed, like an if-then instruction or a while-loop, then a station that pauses through its duration still counts rounds  to stay synchronized, as if it were actively executing the code.

Each station has its private copy of every among the variables used in the pseudocode.
When the values of these copies may vary across the stations, then we add the station's external-observer ``name'' as a subscript of the variable's name to emphasize this, and otherwise, when all the copies of a variable are equal at all times across all the stations then no subscript is used.

We use the notation $\lg x$ to denote the logarithm to the base~$2$ of~$x$.
The base of logarithms is not specified when this is not necessary,  like in an asymptotic notation~$\cO(\log x)$.

\section{Lower bounds and impossibilities}
\label{sec:lower-bounds}

We will show impossibility results to justify methodological approach to naming algorithms we apply, and use lower bounds on performance metrics for such algorithms to argue about the optimality of the algorithms, as developed in subsequent sections.
We begin with an observation, formulated as Proposition~\ref{pro:impossibility-deterministic}, that if a system is sufficiently symmetric then randomness is necessary to break symmetry.
The given argument is standard and is given for completeness sake; see~\cite{Angluin80,Attiya-Ellen-book-2014,FichR03}.


\begin{proposition}
\label{pro:impossibility-deterministic}

There is no deterministic naming algorithm for a synchronous channel with beeping with at least two stations, in which all stations are anonymous, such that it eventually terminates and assigns proper names.
\end{proposition}

\begin{proof}
We argue by contradiction.
Suppose that there exists a deterministic algorithm that eventually terminates with proper names assigned to the anonymous stations.
Let all the stations start initialized to the same initial state.
The following invariant is maintained in each round: the internal states of the stations are all equal.

We proceed to show this invariant by induction on the round numbers.
The base of induction is satisfied by the assumption about the initialization.
For the inductive step, we assume that the stations are in the same state, by the inductive hypothesis.
Then either all of them pause or all of them beep in the next round, as determined by the state.
The effect is that either all of them hear their own beep or all of them pause and hear silence.
This results in the same internal state transition, which completes the inductive step.

When the algorithm eventually terminates, then each station assigns to itself the identifier determined by its state.
The identifier is the same in all stations because their states are the same, by the invariant.
This violates the desired property of names to be distinct, because there are at least two stations with the same name.
\end{proof}

Proposition~\ref{pro:impossibility-deterministic} justifies developing randomized naming algorithms.
We continue with ``entropy'' arguments; see the book by Cover and Thomas~\cite{CoverT-book} for a systematic exposition of information theory.

An execution of a randomized naming algorithm coordinates and translates random bits into names.
This same amount of entropy needs to be processed/communicated on the channel, by the Shannon's noiseless coding theorem.
An analogue of the following Proposition~\ref{pro:lower-bound-on-random-bits} was stated in~\cite{ChlebusMT15} for the model of synchronized processors communicating by reading and writing to shared memory.
The proof is similar, we include it here for completeness sake.


\begin{proposition}
\label{pro:lower-bound-on-random-bits}

If a randomized naming algorithm for a channel with beeping is executed by $n$ anonymous stations and is correct with probability~$p_n$ then it requires the total of $\Omega(n\log n)$ random bits generated with probability at least~$p_n$.
In particular, a Las Vegas naming algorithm uses $\Omega(n\log n)$ random bits almost surely.
\end{proposition}

\begin{proof}
Let us refer to the $n$ anonymous stations by the identifiers known only to an external observer; we call them the \emph{unknown names}.
The names given by an execution of a randomized naming algorithms are referred to as the \emph{given names}.

A permutation of the unknown names is produced by ordering the stations by the given names, with the respective probability.
The algorithm determines a probability distribution on these permutations, conditional on successful naming.
The conditional distribution is uniform because of the symmetry: all the processors execute the same code, without explicit private identifiers, and all use random bits produced by tossing fair coins.
This determines a probability space of $n!$ elementary events, each identified by a permutation of the unknown identifiers.
Each such an elementary event occurs with the conditional probability~$1/n!$.
The Shannon entropy of this random variable equals $\lg (n!)=\Theta(n\log n)$.
Setting on one permutation, by way of assigning names, requires the amount $\lg (n!)=\Theta(n\log n)$ of entropy, by the Shannon's noiseless coding theorem.
The source of this entropy are random bits generated by the stations.
\end{proof}

One round of an execution of a naming algorithm allows the stations that do not transmit to learn at most one bit, because, from the perspective of these stations, a round is either silent or there is a beep.
Intuitively, the running time is proportional to the amount of entropy that is needed to assign names.
This intuition leads to Proposition~\ref{pro:lower-bound-time}.
In its proof, we combine the Shannon's entropy~\cite{CoverT-book} with the Yao's principle~\cite{Yao77}.


\begin{proposition}
\label{pro:lower-bound-time}

A randomized naming algorithm for a beeping channel with $n$ stations  operates in $\Omega(n\log n)$ expected time, when it is either a Las Vegas algorithm or a Monte Carlo algorithm with the probability of error smaller than~$1/2$.
\end{proposition}

\begin{proof}
We apply the Yao's minimax principle to bound the expected time of a randomized algorithm by  the distributional complexity of naming.
We consider Las Vegas algorithms first. 

A randomized  algorithm using strings of random bits generated by stations can be considered as a  deterministic algorithm~$\cD$ on all possible assignments of such (sufficiently long) strings of bits to stations as their inputs.
We consider assignments of strings of bits of equal length, with the uniform distribution among all such assignments of strings of the same length.
On a given assignment of input strings of bits to stations, the deterministic algorithms either assigns proper names or fails to do so.
A failure to assign proper names with some input is interpreted as the randomized algorithm continuing to work with  additional random bits, which comes at an extra time cost.
This is justified by a combination of two factors. 
One is that the algorithm is Las Vegas, so it halts almost surely and with a correct output. 
The other is that the probability to assign a specific finite sequence as a prefix of a used sequence of random bits is positive.
So, if starting with a specific string of bits, as a prefix of a possibly longer needed string, would mean the inability to terminate with a positive probability, then the naming algorithm would not be Las Vegas.

The common length of these input strings is a parameter.
We consider all sufficiently large positive integer values for this parameter such that there exist strings of bits of this length resulting in assignments of proper names.
For a given length of input strings, we remove input assignments that do not result in assignment proper names and consider a uniform distribution of the remaining inputs.
This is the same as the uniform distribution conditional on the algorithm terminating with input strings of bits of a given length.

Let us consider such a deterministic algorithm~$\cD$ assigning names, and using strings of bits at stations as inputs, these strings being of a fixed length, assigned under a uniform distribution for this length, and such that they result in termination.
An execution of this algorithm produces a ``feedback sequence of bits,'' where we translate the feedback from the channel round by round, say, with symbol~$1$ representing a beep and  symbol~$0$ representing silence.
Each such a feedback sequence translates into a specific assignment of names, because this is the only information shared by all the stations.
These feedback sequences also have  a uniform distribution, by the same symmetry argument as used in the proof of Proposition~\ref{pro:lower-bound-on-random-bits}.
The expected length of such a feedback sequence is the expected time of algorithm~$\cD$.
The expected time of algorithm~$\cD$ is therefore at least $\lg n!=\Omega(n\log n)$, by the Shannon's noiseless coding theorem.
We conclude that, by the Yao's principle, the original randomized Las Vegas algorithm has the expected time that is $\Omega(n\log n)$.

A similar argument, by the Yao's principle, applies to a Monte Carlo algorithm that is incorrect with a constant probability smaller than~$1/2$.
The only difference in the argument is that when a given assignment of input string bits does not result in a proper assignment of names, then either the algorithm continues to work with more bits for an extra time, or terminates with error.
\end{proof}

Next, we consider the case when the number of stations~$n$ is unknown.
The following Proposition~\ref{pro:probability-of-error} is about the inevitability of error in such situations, which is due to the fact that we cannot compare the number of randomly assigned names to the (unknown) number of computing/communicating agents.
Intuitively, when two computing/communicating agents generate the same string of bits, then their actions are the same, and so they get the same name assigned.
In other words, we cannot distinguish the case when there is only one such an agent present from  cases when at least two of them work in unison.

Proposition~\ref{pro:probability-of-error} was formulated in~\cite{ChlebusMT15} for the model of synchronous parallel random-access machine, in which processors communicate among themselves by reading from and writing to shared read-write registers.
It applies to a synchronous beeping channel, when we understand actions of stations as either beeping or pausing in a round; we include the proof to make the exposition self-contained.


\begin{proposition}[\cite{ChlebusMT15}]
\label{pro:probability-of-error}

For an unknown number of stations~$n$, if a randomized naming algorithm is executed by $n$ anonymous stations, then an execution is incorrect, in that duplicate names are assigned to different stations, with probability that is at least~$n^{-\Omega(1)}$, assuming that the algorithm uses $\cO(n\log n)$ random bits with probability $1-n^{-\Omega(1)}$.
\end{proposition}

\begin{proof}
Suppose the algorithm uses at most $c n\lg n$ random bits with  probability~$p_n$ when executed by a channel with $n$ stations, for some constant $c>0$.
Then one of these stations uses at most $c \lg n$ bits with  probability~$p_n$, by the pigeonhole principle.

Consider an execution of the same algorithm for $n+1$ stations.
Let us distinguish a station~$v$.
One of the remaining $n$ stations, say~$w$, uses at most $c \lg n$ bits with the probability~$p_n$. 
Station~$v$ generates the same string of bits with probability~$2^{-c\lg n}= n^{-c}$.
The random bits generated by the station $w$ and $v$ are independent.
Therefore, duplicate names occur with probability at  least~$n^{-c}\cdot p_n$.
When the bound $p_n=1-n^{-\Omega(1)}$ holds, then the probability of duplicate names is at least $n^{-c}(1-n^{-\Omega(1)})= n^{-\Omega(1)}$.
\end{proof}

We conclude this section with a fact about impossibility of developing a Las Vegas naming algorithm when the number of stations~$n$ is unknown.


\begin{proposition} 
\label{pro:no-las-vegas}

{\bf (\cite{ChlebusMT15,KuttenOP00})}
There is no Las Vegas naming algorithm for a channel with beeping with at least two stations such that it does not refer to the number of stations.
\end{proposition}

Proposition~\ref{pro:no-las-vegas} was shown by Kutten et al.~\cite{KuttenOP00} for the model of processes communicating by asynchronous processes reading from and writing to shared registers.
This fact depends not that much on asynchrony as on $n$ being unknown, and it was shown by  Chlebus et al.~\cite{ChlebusMT15} to hold for anonymous synchronous shared-memory systems of read-write registers.
The proof given in~\cite{ChlebusMT15}  is general enough to be directly applicable here as well, as both models are synchronous.
Proposition~\ref{pro:no-las-vegas} justifies developing Monte Carlo algorithm for unknown $n$, which we do in Section~\ref{sec:monte-carlo}.

\section{A Las Vegas algorithm}

\label{sec:las-vegas}

We give a Las Vegas naming algorithm for the case when $n$ is known.
The idea is to have stations choose rounds to beep from a segment of integers.
As a convenient probabilistic interpretation, these integers are interpreted as bins, and, after selecting a bin, a ball is placed in the bin.

An execution proceeds by considering the consecutive bins one by one.
First, a bin is verified to be nonempty by hearing that the owners of the balls in the bin beep.
When no beep is heard then the next bin is considered, otherwise the nonempty bin is verified for collisions.
Such a verification is performed by $\cO(\log n)$ consecutive calls of procedure \textsc{Detect-Collision} given in Section~\ref{sec:preliminaries}.
When a collision is not detected then the stations that placed their balls in this bin assign themselves the next available name, otherwise the stations whose balls are in this bin place their balls in a new set of bins.
When each station has a name assigned, we verify if the maximum assigned name is~$n$.
If this is the case then the algorithm terminates, otherwise we repeat this procedure again.

The algorithm is called \textsc{Beep-Naming-LV}, its pseudocode is in Figure~\ref{alg:las-vegas}.
This pseudocode is to be understood that natural inter-processor synchronization occurs, in that when a group of processors participate in executing an instruction while others do not, then those that are not involved wait for the duration of the action by the participants to be completed.
In particular, this applies to the for-loop in Figure~\ref{alg:las-vegas} that is controlled by variable~$i$.
If a beep is heard then this means that some processors will execute the inner for-loop to iterate $\beta\lg n$ times procedure \textsc{Detect-Collision} while the other stations pause for the same amount of time and resume when the for-loop has been completed.
Otherwise, when a beep is not heard, then all the processors increment $i$ and execute the next iteration of the for-loop.


\begin{figure}[t]
\rule{\textwidth}{0.75pt}

\F 
\textbf{Algorithm} \textsc{Beep-Naming-LV}

\rule{\textwidth}{0.75pt}
\begin{center}
\begin{minipage}{\pagewidth}
\begin{description}
\item[\tt repeat]\ 
\begin{description}
\item[\texttt{counter} $\gets 0$] ;  \texttt{range} $\gets n\lg n$ ;  \texttt{name}$_v \gets$ \texttt{null}
\item[\tt repeat]\ 
\begin{description}
\item[\texttt{slot}$_v$] $\gets$\, random number in the interval $[\texttt{left},\texttt{right}]$ 
\item[\tt for] $i\gets$ \texttt{$1$ to range do}
\begin{description}
\item[\tt if] $i = \texttt{slot}_v$ \texttt{then}
\begin{description}
\item[\tt beep]
\item[\tt if] a beep was just heard \texttt{then}
\begin{description}
\item[\tt collision] $\gets$ \texttt{false}
\item[\tt for] $j\gets 1$ \texttt{to} $\beta \lg n$ \texttt{do}

\item[\ \ \ ] \texttt{if} \textsc{Detect-Collision} \ \texttt{then} \ \texttt{collision} $\gets$ \texttt{true}

\item[\tt if not collision then]\   

\item[\ \ \ ]  \texttt{counter} $\gets \texttt{counter} + 1$

\item[\ \ \ ]  \texttt{name}$_v$ $\gets \texttt{counter}$

\end{description}
\end{description}
\end{description}
\item[\tt if] $\texttt{name}_v = \texttt{null then beep}$ 
\hfill
/$\ast$   station $v$ participated in a collision $\ast$/   
\item[\tt if] a beep was just heard \texttt{then}
\begin{description}
\item[\texttt{range}] $\gets (n-\texttt{counter})\lg n$
\end{description}
\end{description}
\item[\tt until] no beep was heard in the previous round
\end{description}
\item[\tt until] \texttt{counter} $= n$
\end{description}
\end{minipage}
\FFF

\rule{\textwidth}{0.75pt}

\parbox{\captionwidth}{\caption{\label{alg:las-vegas}
The pseudocode for a station~$v$.
The number of stations~$n$ is known.
The constant~$\beta>1$ is a parameter determined in the analysis.
Procedure \textsc{Detect-Collision} has its pseudocode in Figure~\ref{fig:procedure-collision}.
The variable \texttt{name} is to store the assigned identifier.
}}
\end{center}
\end{figure}

Algorithm \textsc{Beep-Naming-LV} is analyzed by modeling its executions by a process of throwing balls into bins, which we call the \emph{ball process}.
The process proceeds through stages.
There are $n$ balls in the first stage.
When a stage begins and there are some $i$ balls eligible for the stage then the number of used bins is~$i\lg n$.
Each ball is thrown into a randomly selected bin.
Next, balls that are single in their bins are removed and the remaining balls that participated in collisions advance to the next stage.
The process terminates when no  eligible balls remain.

\begin{lemma}
\label{lem:balls-into-bins-times-log}

The number of times a ball is thrown into a bin during an execution of the ball process that starts  with $n$ balls is at most $3n$ with probability at least $1-e^{-n/4}$.
\end{lemma}

\begin{proof}
In each stage, we throw some $k$ balls into at least $k\lg n$ bins.
The probability that a given ball ends up single in a bin is at least
\[
1-\frac{k}{k\lg n} = 1-\frac{1}{\lg n}
\ ,
\]
which we denote as~$p$.
A ball is thrown repeatedly in consecutive iterations until it lands single in a bin.
Our immediate concern is the number of trials to have all balls end up as singletons in their bins.

Suppose that we perform some $m$ independent Bernoulli trials, each with the probability~$p$ of success, and let $X$ be the number of successes.
Here a success means that a new ball ends up single in its bin.
We show next that taking $m=\Theta(n)$ suffices  to have the inequality $X\ge n$ hold with a large probability.

The expected number of successes is $\mE[X]=\mu=pm$.
We use the Chernoff bound in the form
\begin{equation}
\label{eqn:chernoff}
\Pr(X< (1-\varepsilon)\mu) < e^{-\varepsilon^2\mu/2}
\ ,
\end{equation}
for any $0<\varepsilon<1$; see~\cite{MotwaniR95}.
We want the inequality $(1-\varepsilon)\mu \ge n$ to hold.
Let us set $\varepsilon = \frac{1}{2}$, so that $\mu\ge 2n$ suffices.
This means $pm\ge 2n$, which is extended into the following form:
\begin{equation}
\label{eqn:post-chernoff}
 \Bigl(1-\frac{1}{\lg n}\Bigr)\cdot m\ge 2n
 \ .
\end{equation}
If we choose $m=3n$ then the inequality~\eqref{eqn:post-chernoff} holds for sufficiently large $n$.

The probability that this inequality~\eqref{eqn:post-chernoff} does not hold is estimated from above by~\eqref{eqn:chernoff}.
Here we have that $\mu\ge 2n$ so the right-hand side of~\eqref{eqn:chernoff} is~$e^{-n/4}$.
\end{proof}

We proceed to Theorem~\ref{thm:las-vegas}, which summarizes the desired properties of algorithm \textsc{Beep-Naming-LV}, in particular that it is Las Vegas.
In the proof, we model an execution of the algorithm as the ball process that starts with $n$ balls. 
The main difference between the ball process and the algorithm is that collisions of balls in bins in the ball process are detected with certainty, by the specification of the process, while collisions between tentative names in the algorithm's execution might not be detected with some positive probability.

\begin{theorem}
\label{thm:las-vegas}

Algorithm \textsc{Beep-Naming-LV}, for each $\beta>0$, terminates almost surely and there is no error when it terminates.
For each $a>0$, there exists $\beta>1$ and $c>0$ such that  the algorithm assigns unique names, works in time at most~$c n\lg n$, and  uses at most $c n\lg n$ random bits, all these properties holding with the probability that is at least $1-n^{-a}$, for sufficiently large~$n$.
\end{theorem}

\begin{proof}
Consider an iteration of the main repeat-loop.
An error can occur in this iteration only when there is a collision that is not detected by procedure \textsc{Detect-Collision} in none of its $\beta \lg n$ calls.
Such an error results in duplicate names, so that the number of assigned different names is smaller than~$n$.
The maximum name assigned in an iteration is the value of the variable \texttt{counter}, which has the same value at each station.
The algorithm terminates by having an iteration that produces \texttt{counter} $= n$,  but then there are no repetitions among the names, and so there is no error.

Next, we show that termination is a sure event.
Consider an iteration of the main repeat-loop.
There are $n$ balls and each of them is being thrown repeatedly until either it is not involved in a collision or there is a collision but it is not detected. 
Eventually each ball is left to reside in its bin with probability~$1$.
This means that each iteration ends almost surely.

The following are notations used in analyzing iterations of the main repeat-loop.
Denote by $A$ the event that there is a collision that passes undetected.
The iteration fails to assign proper names if and only if the event~$A$ holds.
Let $B$ be the event that  the total number of throws of balls into bins is at most~$3n$.
Denote by $\neg E$ the complements of an event~$E$. 
We have that $\Pr(\neg B)\le e^{-n/4}$, by Lemma~\ref{lem:balls-into-bins-times-log}.

When a ball lands in a bin then it is verified for a collision $\beta\lg n$ times.
If there is a collision then it passes undetected with a probability that is at most~$n^{-\beta}$.
This is because one call of procedure \textsc{Detect-Collision} detects a collision with a probability that is at least~$\frac{1}{2}$,  by Lemma~\ref{lem:detect-collision}, in which $m=1$ and $k\ge 2$. 

We estimate the probability of the event that an iteration fails to assign proper names, which is the same as of the event~$A$.
This is accomplished as follows:
\begin{eqnarray}
\nonumber
\Pr(A) &= & \Pr(A\cap B)+ \Pr(A\cap \neg B)\\
\nonumber
& = & \Pr(A\mid B) \cdot \Pr(B) + \Pr(A\mid \neg B) \cdot \Pr(\neg B)\\
\nonumber
&\le & \Pr(A\mid B) + \Pr(\neg B)\\
\label{eqn:poly-and-exp}
&\le & 3n\cdot n^{-\beta} + e^{-n/4}
\ ,
\end{eqnarray}
where we used the union bound to obtain the last line~\eqref{eqn:poly-and-exp}.
It follows that at least $i$ iterations are needed with a probability that is at most $(e^{-n/4}+3 n^{1-\beta})^i$, which converges to~$0$ as $i$ grows unbounded, assuming only that $\beta>1$ and $n$ is sufficiently large.

Let us consider modeling an iteration of the main repeat loop.
The event $\neg A\cap B$ occurs when balls are thrown at most $3n$ times and all collisions are detected. 
The probability that the event $\neg A\cap B$ holds can be estimated from below as follows:
\begin{eqnarray}
\nonumber
\Pr(\neg A\cap B)  &=& \Pr(\neg A\mid B) \cdot \Pr(B) \\
\nonumber
&\ge & (1- 3n^{1-\beta}) \cdot (1-e^{-n/4}) \\
\label{eqn:two-products}
&\ge & 1 - 3n^{1-\beta} - e^{-n/4} (1 - 3n^{1-\beta})
\ .
\end{eqnarray}
Bound~\eqref{eqn:two-products} is at least $1-n^{-a}$ for sufficiently large $\beta>1$,  when also $n$ is large enough.

Bound~\eqref{eqn:two-products} holds for the first iteration of the main repeat loop.
So, with a probability that is at least $1-n^{-a}$, the first iteration assigns proper names with at most $3n$ balls thrown in total.
Let us assume that this event occurs.
Then the whole execution takes time at most $c n\lg n$, for a suitably large~$c>0$.
This is because procedure \textsc{Detect-Collision} is executed at most~$3\beta n\lg n$ times, and each of its calls takes two rounds.
One assignment of a value to the variable \texttt{slot} requires $\lg (n\lg n)<2\lg n$ bits, for sufficiently large~$n$.
There are at most $3n$ such assignments, for a total of at most $c n\lg n$ random bits, for a suitably large~$c>0$.
\end{proof}

The optimality of algorithm \textsc{Beep-Naming-LV} follows from the propositions given in Section~\ref{sec:lower-bounds}.
The algorithm runs in the optimal expected time $\cO(n\log n)$, by Proposition~\ref{pro:lower-bound-time}, and it uses the optimum expected number of random bits $\cO(n\log n)$, by Proposition~\ref{pro:lower-bound-on-random-bits}.

\section{A Monte Carlo algorithm}

\label{sec:monte-carlo}

We give a randomized naming algorithm for the case when $n$ is unknown.
In view of Proposition~\ref{pro:no-las-vegas}, no Las Vegas algorithm exists in this case, so we develop a Monte Carlo one.

The algorithm's execution can be interpreted as repeatedly throwing balls into bins and verifying for collisions.
A bin is determined by a string of some $k$ bits.
Each station chooses one such a string randomly.
The algorithm proceeds by repeatedly identifying the string that is smallest lexicographically among those that have not been considered yet.
This is accomplished by procedure \textsc{Next-String} which operates as a concurrent search coordinated by beeps.
Having identified a nonempty bin, all the stations that placed their balls into this bin verify if there is a collision in this bin by calling \textsc{Detect-Collision} a suitably large number of times.
In the case when no collision has been detected, the stations whose balls are in the bin assign themselves the consecutive available name as a temporary one.
This continues until all the balls have been considered.
If no collision has ever been detected in the current stage, then the algorithm terminates and the temporary names are considered as the final assigned names, otherwise the algorithm  proceeds to the next stage. 

Procedure \textsc{Next-String} operates as a radix search.
The goal is to identify the smallest string of bits by examining consecutive bit positions.
The procedure uses two variables \texttt{my-string} and $k$, where $k$ is the length of the  considered bit strings and \texttt{my-string}$_v$ is the string of $k$ bits generated by station~$v$.
The procedure begins by setting each of the $k$ bit positions in variable \texttt{string} to~$1$.
Then these bit positions are considered one by one.
For a given bit position $i$, where $1\le i\le k$, all the stations $v$ that can still have the smallest string and whose bit on position~$i$ in \texttt{my-string}$_v$ is~$0$ do  beep.
This determines the first $i$ bits of the smallest string, because if a beep is heard then the $i$th bit of the smallest string is~$0$ and otherwise it is~$1$.
This is recorded by setting the $i$th bit position in the variable \texttt{string} to the determined bit.
The stations eligible for beeping, if their $i$th bit is~$0$, are those whose strings agree on the first $i-1$ positions with the smallest string.
After all~$k$ bit positions have been considered, the variable \texttt{string} is returned. 

Procedure \textsc{Next-String} has its pseudocode in Figure~\ref{fig:procedure-next-string}.
Its relevant property is summarized in Lemma~\ref{lem:next-string}.


\begin{figure}[t]
\rule{\textwidth}{0.75pt}

\F 
\textbf{Procedure} \textsc{Next-String}

\rule{\textwidth}{0.75pt}
\begin{center}
\begin{minipage}{\pagewidth}
\begin{description}
\item[\tt string] $\gets$  a string of $k$ bit positions, with each of them set to $1$
\item[\tt for] $i\gets 1$ \texttt{to} $k$ \texttt{do}
\begin{description}
\item[\tt if]    (\texttt{my-string}$_v$ matches \texttt{string} on the first $i-1$ bit positions)
\item[\tt \ \ \ \ \ \ \ \ \ \ \ \ \ \ \ \ and] (the $i$th bit of \texttt{my-string}$_v$ is~$0$) 
\begin{description}
\item[\texttt{then}] \texttt{beep}
\end{description}
\item[\texttt{if}] a beep was heard in the previous round \texttt{then}
\begin{description}
\item[\rm set] the $i$th bit of \texttt{string} to~$0$ 
\end{description} 
\end{description}
\item[\tt return] (\texttt{string})
\end{description}
\end{minipage}
\FFF

\rule{\textwidth}{0.75pt}

\parbox{\captionwidth}{\caption{\label{fig:procedure-next-string}
The pseudocode for a station $v$.
This procedure is used by algorithm \textsc{Beep-Naming-MC}.
The variables \texttt{my-string} and $k$ are the same as those in the pseudocode in Figure~\ref{alg:monte-carlo}.
}}
\end{center}
\end{figure}


\begin{lemma}
\label{lem:next-string}

Procedure \textsc{Next-String} returns the smallest lexicographically string  among the non-null string values of the private copies of the variable \texttt{my-string}.
\end{lemma}

\begin{proof}
The string that is output is obtained by processing all the input strings \texttt{my-string} through consecutive bit positions.
We show the invariant that after $i$ bits have been considered, for $0\le i\le k$, then the bits on these positions make the prefix of the first $i$ bits of the smallest string. 

The invariant is shown by induction on~$i$.
When $i=1$ then the bits on previously considered positions make an empty string, as no positions have been considered yet, and the empty string is a prefix of the smallest string.
Suppose that the invariant holds for all $i$ such that $0\le i<k$, and consider the stations whose variable \texttt{my-string} has the same bits on these first $i$ positions as variable \texttt{string}.
This set includes the station~$v$ with the smallest \texttt{my-string} by the inductive hypothesis.
If the bit on the $(i+1)$st position of \texttt{my-string}$_v$ is~$0$ then $v$ beeps and \texttt{string} has its bit on position $i+1$ set to~$0$.
Otherwise there is no station with $0$ on the $(i+1)$st position of \texttt{my-string}$_v$, because \texttt{my-string}$_v$ is smallest.
Then there is no beep and $1$ at position $i+1$ in \texttt{string} is not modified.
This completes the proof of the invariant.

The procedure \textsc{Next-String} terminates after $k$ bit positions have been processed.
The proved invariant, applied for $i=k$, means that the final value of the variable \texttt{string} and the smallest \texttt{my-string}$_v$, for all the nodes~$v$, are identical.
\end{proof}

The naming algorithm we develop next is called \textsc{Beep-Naming-MC}.
Its pseudocode is in Figure~\ref{alg:monte-carlo}.


\begin{figure}[t]
\rule{\textwidth}{0.75pt}

\F 
\textbf{Algorithm} \textsc{Beep-Naming-MC}

\rule{\textwidth}{0.75pt}
\begin{center}
\begin{minipage}{\pagewidth}
\begin{description}
\item[$k\gets 1$]\ 
\item[\tt repeat]\ 
\begin{description}
\item[$k \gets 2k$] 
\item[\texttt{collision}] $\gets$ \texttt{false}
\item[\texttt{counter}] $\gets 0$ 
\item[\tt my-string$_v$]$\gets$ a random string of $k$ bits
\item[\tt repeat]\ 
\begin{description}
\item[\tt if]  \texttt{my-string$_v$} $\ne$ \texttt{null then smallest-string} $\gets$ \textsc{Next-String} 
\item[\tt if] \texttt{my-string$_v$} $= \texttt{smallest-string}$ \texttt{then}
\begin{description}

\item[\tt for] $i\gets 1$ \texttt{to} $\beta\cdot k$ \texttt{do}
\item[\ \ \ ] \texttt{if} \textsc{Detect-Collision} \ \texttt{then} \ \texttt{collision} $\gets$ \texttt{true}

\item[\tt if not collision then] \ 
\begin{description}
\item[\tt counter] $\gets \texttt{counter} + 1$
 \item[\tt name$_v$] $\gets \texttt{counter}$

\end{description}
  \item[\tt my-string$_v$] $\gets \texttt{null}$
\end{description}
\item[\tt if] $\texttt{my-string}_v \ne \texttt{null}$ \texttt{then beep}
\end{description}
\item[\tt until] no beep was heard in the previous round
\end{description}
\item[\tt until] \texttt{not collision}
\end{description}
\end{minipage}
\FFF

\rule{\textwidth}{0.75pt}

\parbox{\captionwidth}{\caption{\label{alg:monte-carlo}
The pseudocode for a station~$v$.
Constant $\beta>0$ is an integer parameter determined in the analysis.
The pseudocode of procedure \textsc{Detect-Collision} is given in Figure~\ref{fig:procedure-collision} and the pseudocode of procedure \textsc{Next-String} is in Figure~\ref{fig:procedure-next-string}.
The variable \texttt{name} is to store the assigned identifier.
}}
\end{center}
\end{figure}

The algorithm proceeds through stages, where a stage is implemented by an iteration of the main repeat-loop in the pseudocode.
The number of bins in the $i$th stage is $2^k$ where $k=2^i$.
The variable~$k$ is doubled in the beginning of each iteration of the main loop.
During a stage, first the next bin with a ball is identified by calling procedure \textsc{Next-String}.
Next, this bin is verified for collisions by calling procedure \textsc{Detect-Collision} $\beta k$ times, for a constant $\beta>0$, which is a parameter to be settled in analysis.
During such a verification, only the stations whose balls are in this bin do participate, while the remaining stations pause for the same amount of time.
This is possible to achieve because the number of times procedure \textsc{Detect-Collision} is called is $\beta\cdot k$, and each call takes two rounds.

Theorem~\ref{thm:monte-carlo} summarizes the properties of algorithm  \textsc{Beep-Naming-MC}.
In particular, that it is a Monte Carlo algorithm with a suitably small probability of error.

\begin{theorem}
\label{thm:monte-carlo}

Algorithm \textsc{Beep-Naming-MC}, for each $\beta>0$, terminates almost surely.
For each $a>0$, there exists $\beta>0$ and $c>0$ such that  the algorithm assigns unique names, works in time at most~$c n\lg n$, and  uses at most $c n\lg n$ random bits, all these properties holding with a probability that is at least $1-n^{-a}$.
\end{theorem}

\begin{proof}
We interpret an iteration of the outer repeat-loop as a stage in a process of throwing $n$ balls into $2^k$ bins and verifying $\beta k$ times for collisions. 
The string selected by a station is the name of the bin.
When at least one collision is detected during an iteration of the outer repeat-loop then variable $k$ gets doubled and another iteration is performed.
An error occurs when there is a collision but it is not detected.

Next we estimate from above the probability of not detecting a collision.
To this end, we consider two cases, depending on which of the inequalities $2^k<n$ or $2^k\ge n$ hold, for a given~$k$.

In the first case, when $2^k<n$, collisions occur with certainty, by the pigeonhole principle.
Let $m$ be the number of occupied bins.
This results in $m\le 2^k$ verifications performed, one for each bin,  where procedure \textsc{Detect-Collision} is called~$\beta k$ times per verification.
By Lemma~\ref{lem:detect-collision}, the probability of not detecting a collision, with just one call of \textsc{Detect-Collision} occurring in each of these verifications, is at most
\[
2^{-n+m}\le 2^{-n+2^k}
\ .
\] 
When $\beta k$ calls of \textsc{Detect-Collision} occur in each verification, as is the case by the design of the algorithm, the probability of not detecting a collision is at most
\[
2^{(-n+2^k)\beta k}
\ .
\]
Intuitively at this point, since $2^k<n$, this probability is maximized for $n=2^k+1$ and it is about $2^{-\beta k}\approx  n^{-\beta}$, as $k\approx \lg n$.

A precise argument to obtain an estimate is by considering two sub-cases, and is as follows.
If it is the sub-case that $2^k< n/2$, then 
\[
2^{(-n+2^k)\beta k}=2^{-\Omega (n)}
\ .
\]
Otherwise, when it is the sub-case that $n>2^k\ge n/2$, then $n-2^k\ge 1$ and $k\ge \lg n-1$, so that we obtain the following estimate:
\[
2^{(-n+2^k)\beta k}\ge 2^{-\beta (\lg n -1)}\ge 2^\beta n^{-\beta}
\ .
\]
We can conclude this sub-case with the following two estimates: $2^{-\Omega (n)}$ and $2^\beta n^{-\beta}$, of which the latter is larger, for sufficiently large~$n$. 
It is sufficient to take $\beta>a$, as then the inequality $2^\beta n^{-\beta}<n^{-a}$ holds for sufficiently large~$n$.

The second case occurs when $2^k\ge n$.
This implies that $k\ge \lg n$.
When a collision occurs in a bin, then it is verified by  at least $\beta \lg n$ calls of procedure \textsc{Detect-Collision}.
This gives the probability of not detecting one such a collision  that is at most $n^{-\beta}$.
Multiple bins with collisions make the probability of not detecting any of them even smaller.
Now it is enough to take $\beta>a$, as then $n^{-\beta}<n^{-a}$ holds.

This completes estimating the probability of error by at most $n^{-a}$, for sufficiently large $\beta$, and all correspondingly large~$n$.

Next, we estimate the probability that the running time is $\cO(n\log n)$.
Let us consider a stage with sufficiently many bins, say, when $k=d\lg n$ for $d>2$.
Then the number of bins is~$2^k=n^d$.
The probability that there is no collision at all in this stage is at least 
\begin{equation}
\label{eqn:d-log-n}
\Bigl(1-\frac{n}{n^d}\Bigr)^n\ge 1-\frac{n}{n^{d-1}}= 1-n^{-d+2}
\ .
\end{equation}
Choosing $d= a+2$ we obtain that the algorithm terminates by the iteration of the outer repeat-loop when $k=d\lg n$ with a probability that is at least $1-n^{-a}$.

One iteration of the outer repeat loop, for some~$k$, is proportional to $k\cdot n$.
The total time spent up to and including $k=d\lg n$ is proportional to
\begin{equation}
\label{eqn:time-and-bits}
\sum_{i=1}^{\lg((a+2)\lg n)} 2^i \cdot n\le n\cdot 2 (a+2)\lg n =\cO(n\log n)
\end{equation}
with a probability that is at least $1-n^{-a}$.

The number of bits generated up to and including the iteration for $k=d\lg n$ is also proportional to~\eqref{eqn:time-and-bits}.
This is because the number of bits generated in one iteration of the main repeat-loop is proportional to $k\cdot n$, similarly as the running time.

To show that the algorithm terminates almost surely, it is sufficient to demonstrate that the probability of a collision converges to zero with $k$ increasing.
The probability of no collision for $k=d\lg n$ is at most $n^{-d+2}$, by~\eqref{eqn:time-and-bits}.
If $k$ grows to infinity then $d=k/\lg n$ increases to infinity as well, and then $n^{-d+2}$ converges  to~$0$ as a function of~$d$.
\end{proof}

Algorithm \textsc{Beep-Naming-MC} is optimal with respect to the following performance measures: the expected running time $\cO(n\log n)$, by Proposition~\ref{pro:lower-bound-time}, the expected number of random bits used~$\cO(n\log n)$, by Proposition~\ref{pro:lower-bound-on-random-bits}, and the probability of error, as determined by the number of used  bits, by Proposition~\ref{pro:probability-of-error}.

\section{Conclusion}

\label{sec:conclusion}

We considered a channel in which a synchronized beeping is the only means of communication.
We showed that unique names can be assigned to the anonymous stations by randomized algorithms.
The algorithms are either Las Vegas or Monte Carlo, depending on whether the number of stations~$n$ is known or not, respectively.
The performance characteristics of the two algorithms, such as the running time, the number of random bits used, and the probability of error, are proved to be asymptotically optimal.

The algorithms we developed rely in an essential manner on synchronization of the channel.
It would be interesting to consider an anonymous asynchronous beeping channel and investigate how to assign names to stations in such a communication environment.


\bibliographystyle{abbrv}

\bibliography{beeping}

\begin{thebibliography}{10}

\bibitem{AfekABCHK13}
Y.~Afek, N.~Alon, Z.~Bar{-}Joseph, A.~Cornejo, B.~Haeupler, and F.~Kuhn.
\newblock Beeping a maximal independent set.
\newblock {\em Distributed Computing}, 26(4):195--208, 2013.

\bibitem{AfekABHBB11}
Y.~Afek, N.~Alon, O.~Barad, E.~Hornstein, N.~Barkai, and Z.~Bar-Joseph.
\newblock A biological solution to a fundamental distributed computing problem.
\newblock {\em Science}, 331(6014):183--185, 2011.

\bibitem{AfekM94}
Y.~Afek and Y.~Matias.
\newblock Elections in anonymous networks.
\newblock {\em Information and Computation}, 113(2):312--330, 1994.

\bibitem{Angluin80}
D.~Angluin.
\newblock Local and global properties in networks of processors (extended
  abstract).
\newblock In {\em Proceedings of the $12$th {ACM} Symposium on Theory of
  Computing (STOC)}, pages 82--93, 1980.

\bibitem{AngluinADFP2006}
D.~Angluin, J.~Aspnes, Z.~Diamadi, M.~J. Fischer, and R.~Peralta.
\newblock Computation in networks of passively mobile finite-state sensors.
\newblock {\em Distributed Computing}, 18(4):235--253, 2006.

\bibitem{AngluinAFJ08}
D.~Angluin, J.~Aspnes, M.~J. Fischer, and H.~Jiang.
\newblock Self-stabilizing population protocols.
\newblock {\em ACM Transactions on Autonomous and Adaptive Systems}, 3(4),
  2008.

\bibitem{AntaMM13}
A.~F. Anta, M.~A. Mosteiro, and J.~R. Mu{\~n}oz.
\newblock Unbounded contention resolution in multiple-access channels.
\newblock {\em Algorithmica}, 67(3):295--314, 2013.

\bibitem{AspnesR07}
J.~Aspnes and E.~Ruppert.
\newblock An introduction to population protocols.
\newblock {\em Bulletin of the {EATCS}}, 93:98--117, 2007.

\bibitem{Attiya-Ellen-book-2014}
H.~Attiya and F.~Ellen.
\newblock {\em Impossibility Results for Distributed Computing}.
\newblock Synthesis Lectures on Distributed Computing Theory. Morgan {\&}
  Claypool Publishers, 2014.

\bibitem{AttiyaSW88}
H.~Attiya, M.~Snir, and M.~K. Warmuth.
\newblock Computing on an anonymous ring.
\newblock {\em Journal of the {ACM}}, 35(4):845--875, 1988.

\bibitem{BrandesKKPW16}
P.~Brandes, M.~Kardas, M.~Klonowski, D.~Paj\k{a}k, and R.~Wattenhofer.
\newblock Approximating the size of a radio network in beeping model.
\newblock In {\em Proceedings of the $23$rd International Colloquium on
  Structural Information and Communication Complexity (SIROCCO)}, volume 9988
  of {\em Lecture Notes in Computer Science}, pages 358--373. Springer, 2016.

\bibitem{BuhrmanPSV-2006}
H.~Buhrman, A.~Panconesi, R.~Silvestri, and P.~Vitanyi.
\newblock On the importance of having an identity or, is consensus really
  universal?
\newblock {\em Distributed Computing}, 18(3):167--176, 2006.

\bibitem{Capetanakis79}
J.~Capetanakis.
\newblock Tree algorithms for packet broadcast channels.
\newblock {\em {IEEE} Transactions on Information Theory}, 25(5):505--515,
  1979.

\bibitem{Chlebus-randomized-radio-chapter-2001}
B.~S. Chlebus.
\newblock Randomized communication in radio networks.
\newblock In P.~M. Pardalos, S.~Rajasekaran, J.~H. Reif, and J.~D.~P. Rolim,
  editors, {\em Handbook of Randomized Computing}, volume~I, pages 401--456.
  Kluwer Academic Publishers, 2001.

\bibitem{ChlebusDK16}
B.~S. Chlebus, G.~De~Marco, and D.~R. Kowalski.
\newblock Scalable wake-up of multi-channel single-hop radio networks.
\newblock {\em Theoretical Computer Science}, 615:23--44, 2016.

\bibitem{ChlebusMT15}
B.~S. Chlebus, G.~De~Marco, and M.~Talo.
\newblock Anonymous processors with synchronous shared memory.
\newblock {\em CoRR}, abs/1507.02272, 2015.

\bibitem{ChlebusDT-FI17}
B.~S. Chlebus, G.~De~Marco, and M.~Talo.
\newblock Naming a channel with beeps.
\newblock {\em Fundamenta Informatica}, 153(3):199--219, 2017.

\bibitem{ChlebusDT-OPODIS17}
B.~S. Chlebus, G.~De~Marco, and M.~Talo.
\newblock Anonymous processors with synchronous shared memory: {M}onte {C}arlo
  algorithms.
\newblock In {\em Proceedings of the $21$st International Conference on
  Principles of Distributed Systems (OPODIS 2017)}, volume~95 of {\em Leibniz
  International Proceedings in Informatics (LIPIcs)}, pages 15:1--15:17.
  Schloss Dagstuhl--Leibniz-Zentrum fuer Informatik, 2018.

\bibitem{ChlebusGKR-ICALP05}
B.~S. Chlebus, L.~G{\k a}sieniec, D.~R. Kowalski, and T.~Radzik.
\newblock On the wake-up problem in radio networks.
\newblock In {\em Proceedings of the $32$nd International Colloquium on
  Automata, Languages and Programming (ICALP)}, volume 3580 of {\em Lecture
  Notes in Computer Science}, pages 347--359. Springer, 2005.

\bibitem{CidonS88}
I.~Cidon and M.~Sidi.
\newblock Conflict multiplicity estimation and batch resolution algorithms.
\newblock {\em IEEE Transactions on Information Theory}, 34(1):101--110, 1988.

\bibitem{CornejoK10}
A.~Cornejo and F.~Kuhn.
\newblock Deploying wireless networks with beeps.
\newblock In {\em Proceedings of the $24$th International Symposium on
  Distributed Computing (DISC)}, volume 6343 of {\em Lecture Notes in Computer
  Science}, pages 148--162. Springer, 2010.

\bibitem{CoverT-book}
T.~M. Cover and J.~A. Thomas.
\newblock {\em Elements of Information Theory}.
\newblock Wiley, $2$nd edition, 2006.

\bibitem{CzumajD15}
A.~Czumaj and P.~Davies.
\newblock Communicating with beeps.
\newblock In {\em Proceedings of the $19$th International Conference on
  Principles of Distributed Systems (OPODIS)}, volume~46 of {\em LIPIcs}, pages
  30:1--30:16. Schloss Dagstuhl - Leibniz-Zentrum fuer Informatik, 2015.

\bibitem{DeMarcoK10}
G.~De~Marco and D.~R. Kowalski.
\newblock Towards power-sensitive communication on a multiple-access channel.
\newblock In {\em Proceedings of the 2010 International Conference on
  Distributed Computing Systems (ICDCS)}, pages 728--735. {IEEE} Computer
  Society, 2010.

\bibitem{DeMarcoK15}
G.~De~Marco and D.~R. Kowalski.
\newblock Fast nonadaptive deterministic algorithm for conflict resolution in a
  dynamic multiple-access channel.
\newblock {\em {SIAM} Journal on Computing}, 44(3):868--888, 2015.

\bibitem{DeMarcoK17}
G.~De~Marco and D.~R. Kowalski.
\newblock Contention resolution in a non-synchronized multiple access channel.
\newblock {\em Theoretical Computer Science}, 689:1--13, 2017.

\bibitem{DeMarcoPS07}
G.~De~Marco, M.~Pellegrini, and G.~Sburlati.
\newblock Faster deterministic wakeup in multiple access channels.
\newblock {\em Discrete Applied Mathematics}, 155(8):898--903, 2007.

\bibitem{DegesysRPN07}
J.~Degesys, I.~Rose, A.~Patel, and R.~Nagpal.
\newblock {DESYNC:} self-organizing desynchronization and {TDMA} on wireless
  sensor networks.
\newblock In {\em Proceedings of the $6$th International Conference on
  Information Processing in Sensor Networks (IPSN)}, pages 11--20. ACM, 2007.

\bibitem{EgeciogluS94}
{\"O}.~E{\u{g}}ecio{\u{g}}lu and A.~K. Singh.
\newblock Naming symmetric processes using shared variables.
\newblock {\em Distributed Computing}, 8(1):19--38, 1994.

\bibitem{EmekW13}
Y.~Emek and R.~Wattenhofer.
\newblock Stone age distributed computing.
\newblock In {\em Proceedings of the 2013 ACM Symposium on Principles of
  Distributed Computing (PODC)}, pages 137--146, 2013.

\bibitem{FichHRW85}
F.~E. Fich, F.~{Meyer auf der Heide}, P.~Ragde, and A.~Wigderson.
\newblock One, two, three\dots infinity: Lower bounds for parallel computation.
\newblock In {\em Proceedings of the $17$th {ACM} Symposium on Theory of
  Computing (STOC)}, pages 48--58, 1985.

\bibitem{FichR03}
F.~E. Fich and E.~Ruppert.
\newblock Hundreds of impossibility results for distributed computing.
\newblock {\em Distributed Computing}, 16(2-3):121--163, 2003.

\bibitem{FluryW10}
R.~Flury and R.~Wattenhofer.
\newblock Slotted programming for sensor networks.
\newblock In {\em Proceedings of the $9$th International Conference on
  Information Processing in Sensor Networks (IPSN)}, pages 24--34. {ACM}, 2010.

\bibitem{ForsterSW14}
K.~F{\"{o}}rster, J.~Seidel, and R.~Wattenhofer.
\newblock Deterministic leader election in multi-hop beeping networks -
  (extended abstract).
\newblock In {\em Proceedings of the $28$th International Symposium on
  Distributed Computing (DISC)}, volume 8784 of {\em Lecture Notes in Computer
  Science}, pages 212--226. Springer, 2014.

\bibitem{GasieniecPP-JDM01}
L.~G{\k a}sieniec, A.~Pelc, and D.~Peleg.
\newblock The wakeup problem in synchronous broadcast systems.
\newblock {\em SIAM Journal on Discrete Mathematics}, 14(2):207--222, 2001.

\bibitem{GhaffariH13}
M.~Ghaffari and B.~Haeupler.
\newblock Near optimal leader election in multi-hop radio networks.
\newblock In {\em Proceedings of the $24$th {ACM-SIAM} Symposium on Discrete
  Algorithms (SODA)}, pages 748--766. {SIAM}, 2013.

\bibitem{GhaffariHK15}
M.~Ghaffari, B.~Haeupler, and M.~Khabbazian.
\newblock Randomized broadcast in radio networks with collision detection.
\newblock {\em Distributed Computing}, 28(6):407--422, 2015.

\bibitem{GilbertN15}
S.~Gilbert and C.~C. Newport.
\newblock The computational power of beeps.
\newblock In {\em Proceedings of the $29$th International Symposium on
  Distributed Computing (DISC)}, volume 9363 of {\em Lecture Notes in Computer
  Science}, pages 31--46. Springer, 2015.

\bibitem{GoussevskaiaPW10-survey}
O.~Goussevskaia, Y.~A. Pignolet, and R.~Wattenhofer.
\newblock Efficiency of wireless networks: Approximation algorithms for the
  physical interference model.
\newblock {\em Foundations and Trends in Networking}, 4(3):313--420, 2010.

\bibitem{GreenbergFL87}
A.~G. Greenberg, P.~Flajolet, and R.~E. Ladner.
\newblock Estimating the multiplicities of conflicts to speed their resolution
  in multiple access channels.
\newblock {\em Journal of the ACM}, 34(2):289--325, 1987.

\bibitem{Hayes78}
J.~F. Hayes.
\newblock An adaptive technique for local distribution.
\newblock {\em IEEE Transactions on Communications}, 26(8):1178--1186, 1978.

\bibitem{HounkanliP16}
K.~Hounkanli and A.~Pelc.
\newblock Asynchronous broadcasting with bivalent beeps.
\newblock In {\em Proceedings of the $23$rd International Colloquium on
  Structural Information and Communication Complexity (SIROCCO)}, volume 9988
  of {\em Lecture Notes in Computer Science}, pages 291--306. Springer, 2016.

\bibitem{HuangM13}
B.~Huang and T.~Moscibroda.
\newblock Conflict resolution and membership problem in beeping channels.
\newblock In {\em Proceedings of the $27$th International Symposium on
  Distributed Computing (DISC)}, volume 8205 of {\em Lecture Notes in Computer
  Science}, pages 314--328. Springer, 2013.

\bibitem{JeavonsSX2016}
P.~Jeavons, A.~Scott, and L.~Xu.
\newblock Feedback from nature: simple randomised distributed algorithms for
  maximal independent set selection and greedy colouring.
\newblock {\em Distributed Computing}, 29(5):377--393, 2016.

\bibitem{JurdzinskiK-encyclopedia}
T.~Jurdzi{\'n}ski and D.~R. Kowalski.
\newblock Distributed randomized broadcasting in wireless networks under the
  {SINR} model.
\newblock In M.-Y. Kao, editor, {\em Encyclopedia of Algorithms}. Springer US,
  2014.

\bibitem{KranakisK97}
E.~Kranakis and D.~Krizanc.
\newblock Distributed computing on anonymous hypercube networks.
\newblock {\em Journal of Algorithms}, 23(1):32--50, 1997.

\bibitem{KuttenOP00}
S.~Kutten, R.~Ostrovsky, and B.~Patt{-}Shamir.
\newblock The {Las-Vegas} processor identity problem ({H}ow and when to be
  unique).
\newblock {\em Journal of Algorithms}, 37(2):468--494, 2000.

\bibitem{LiptonP-1990}
R.~J. Lipton and A.~Park.
\newblock The processor identity problem.
\newblock {\em Information Processing Letters}, 36(2):91--94, 1990.

\bibitem{MichailCS11}
O.~Michail, I.~Chatzigiannakis, and P.~G. Spirakis.
\newblock {\em New Models for Population Protocols}.
\newblock Synthesis Lectures on Distributed Computing Theory. Morgan {\&}
  Claypool Publishers, 2011.

\bibitem{MotskinRSG09}
A.~Motskin, T.~Roughgarden, P.~Skraba, and L.~J. Guibas.
\newblock Lightweight coloring and desynchronization for networks.
\newblock In {\em Proceedings of the $28$th {IEEE} International Conference on
  Computer Communications (INFOCOM)}, pages 2383--2391, 2009.

\bibitem{MotwaniR95}
R.~Motwani and P.~Raghavan.
\newblock {\em Randomized Algorithms}.
\newblock Cambridge University Press, 1995.

\bibitem{NavlakhaB2014}
S.~Navlakha and Z.~Bar-Joseph.
\newblock Distributed information processing in biological and computational
  systems.
\newblock {\em Communications of the ACM}, 58(1):94--102, 2014.

\bibitem{PanconesiPTV97}
A.~Panconesi, M.~Papatriantafilou, P.~Tsigas, and P.~M.~B. Vit{\'{a}}nyi.
\newblock Randomized naming using wait-free shared variables.
\newblock {\em Distributed Computing}, 11(3):113--124, 1998.

\bibitem{SchieberS94}
B.~Schieber and M.~Snir.
\newblock Calling names on nameless networks.
\newblock {\em Information and Computation}, 113(1):80--101, 1994.

\bibitem{SchmidW06}
S.~Schmid and R.~Wattenhofer.
\newblock Algorithmic models for sensor networks.
\newblock In {\em Proceedings of the $20$th International Parallel and
  Distributed Processing Symposium (IPDPS)}. {IEEE}, 2006.

\bibitem{TsybakovM1978}
B.~S. Tsybakov and V.~A. Mikhailov.
\newblock Free synchronous packet access in a broadcast channel with feedback.
\newblock {\em Problems of Information Transmission}, 14(4):259--280, 1978.

\bibitem{Yao77}
A.~C. Yao.
\newblock Probabilistic computations: Toward a unified measure of complexity.
\newblock In {\em Proceedings of the $18$th Symposium on Foundations of
  Computer Science (FOCS)}, pages 222--227. {IEEE} Computer Society, 1977.

\bibitem{YuJYLC15}
J.~Yu, L.~Jia, D.~Yu, G.~Li, and X.~Cheng.
\newblock Minimum connected dominating set construction in wireless networks
  under the beeping model.
\newblock In {\em Proceedings of the 2015 {IEEE} Conference on Computer
  Communications (INFOCOM)}, pages 972--980, 2015.

\end{thebibliography}

\end{document}